\documentclass[a4paper,leqno]{CAIMclass}
\newcommand{\DOI}{XXX} 
\newcommand{\caimeditor}{Associated Editor}
\newcommand{\yyyyrec}{2014}
\newcommand{\mmrec}{01}
\newcommand{\ddrec}{28}
\newcommand{\yyyyacc}{YYYY}
\newcommand{\mmacc}{MM}
\newcommand{\ddacc}{DD}
\newcommand{\yyyypub}{YYYY}
\newcommand{\mmpub}{MM}
\newcommand{\ddpub}{DD}
\begin{document}

\title{Nonlinear time-fractional dispersive equations}

\author{Pietro Artale Harris$^1$, Roberto Garra$^1$}

\address{$^1$Dipartimento di Scienze di Base e Applicate per l'Ingegneria, \\``Sapienza'' Universit\`a di Roma\\
Via A. Scarpa 16, 00161 Rome, Italy\\
pietro.artaleharris@sbai.uniroma1.it\\
roberto.garra@sbai.uniroma1.it}

\begin{flushright}
\textit{Dedicated to Professor Francesco Mainardi\\ on the occasion of his retirement}
\vspace{0.7truecm}
\end{flushright}

\begin{abstract}

In this paper we study some cases of time-fractional nonlinear
dispersive equations (NDEs) involving Caputo derivatives, by means
of the invariant subspace method. This method allows to find exact
solutions to nonlinear time-fractional partial differential
equations by separating variables. We first consider a third order
time-fractional NDE that admits a four-dimensional invariant
subspace and we find a similarity solution. We also study a fifth
order NDE. In this last case we find a solution involving
Mittag-Leffler functions. We finally observe that the invariant
subspace method permits to find explicit solutions for a wide class
of nonlinear dispersive time-fractional equations.

\end{abstract}

\keywords{Nonlinear dispersion, Invariant subspace method, Fractional differential equations.}
\AMScode{26A33, 34A08, 35R11.}

\section{Introduction.}

In the present paper some exact solutions to time-fractional
Nonlinear Dispersive Equations (NDEs) by means of the invariant
subspace method (see e.g. \cite{Gala}) are studied. In more details,
we generalize to the fractional case some of the results presented
in \cite{Gala, Gala1}. We discuss the peculiarity of the fractional
cases, where the solutions can not be obtained as a trivial
generalization of the ordinary one. In this context, we generalize
to the time-fractional case the fifth order NDEs firstly considered
by Dey in \cite{Dey}. In this work the author considered the
following family of nonlinear partial differential equations, termed
$K(m,n,p)$
\begin{equation*}
\frac{\partial u}{\partial t}=\nu\frac{\partial^5 u^p}{\partial
x^5}+\beta\frac{\partial^3 u^n}{\partial x^3}+\gamma\frac{\partial
u^m}{\partial x}, \quad m,n,p>1.
\end{equation*}
Thus, for the fractional case, we consider a new family of equations
involving fractional derivatives in the Caputo sense, termed
$K_{\alpha}(m,n,p)$
\begin{equation}\label{K}
\frac{\partial^{\alpha} u}{\partial t^{\alpha}}=\nu\frac{\partial^5
u^p}{\partial x^5}+\beta\frac{\partial^3 u^n}{\partial
x^3}+\gamma\frac{\partial u^m}{\partial x}, \quad \alpha \in (0,1).
\end{equation}
We underline that now $m,n,p\in
\mathbb{R}^+$. \\
In more details we first consider the third order
time-fractional NDE $K_{\alpha}(0,2,0)$
\begin{equation}
\frac{\partial^\alpha u}{\partial
t^\alpha}=\frac{\partial^3}{\partial x^3}\left(\frac{u^2}{2}\right),
\end{equation}
that in the ordinary case $\alpha=1$ was deeply studied in
\cite{Gala1}. By means of the invariant subspace method, we find a
generalization of the similarity solution.\\
The second case under consideration is the fifth-order NDE equation
$K_{\alpha}(2,2,2)$
\begin{equation}
\frac{\partial^\alpha u}{\partial t^\alpha}=\nu\frac{\partial^5
u^2}{\partial x^5}+\beta\frac{\partial^3 u^2}{\partial
x^3}+\gamma\frac{\partial u^2}{\partial x},
\end{equation}
whose  ordinary case  was firstly studied in the framework of compacton solutions of NDEs (see \cite{Dey, Gala}).\\
The main aim of this paper is to show the utility of the invariant
subspace method in order to find exact solutions of time-fractional
NDEs. In particular, also in the light of the literature on
time-fractional NDEs (see e.g. \cite{Guo, Odibat, Odibat1}), this
method provides rigorous and effective tools to generalize compact
soliton solutions and similarity solutions for time-fractional
NDEs.\\
This paper is organized as follows: in Section \ref{sec:2}
we recall some preliminary definitions and results about fractional
derivatives and invariant subspace method. In Section \ref{sec:3} we
study the time-fractional third order NDE while in Section
\ref{sec:4} the time-fractional fifth order NDE is considered.
Finally we suggest further work to be done in the final Section
\ref{sec:5}.

\section{Preliminaries.}
\label{sec:2}
\subsection{Generalities about fractional calculus.}
\label{ssec:1}
Here we recall definitions and basic results about fractional
calculus, for more details we refer to \cite{Kilbas, ma, mai1, Podl}.

\bigskip

Let $\gamma$ be a positive real number. The Riemann-Liouville fractional
integral is defined by
\begin{equation}
J^{\gamma}_t f(t) =
\frac{1}{\Gamma(\gamma)}\int_0^{t}(t-\tau)^{\gamma-1}f(\tau) d\tau,
\label{riemann-l}
\end{equation}
where
$$\Gamma(\gamma)= \int_0^{+\infty}x^{\gamma-1}e^{-x}dx,$$
is the Euler Gamma function. Note that, by definition, $J^0_t f(t)= f(t)$. Moreover it satisfies the semigroup property, i.e. $J_t^{\alpha}J_t^{\beta} f(t)= J_t^{\alpha+\beta}f(t)$.\\
There are different definitions of fractional derivative (see e.g \cite{Podl}). In this paper we use the fractional derivatives in
the sense of Caputo. Hereafter we denote by $AC^n([0,t])$,
$n\in\mathbb{N}$, the class of functions $f(x)$ which are
continuously differentiable in $[0,t]$ up to order $n-1$ and with
$f^{(n-1)}(x)\in AC([0,t])$. We recall the following Theorem
(\cite[pagg. 92-93]{Kilbas})
\begin{theorem}
Let $m-1 < \gamma< m$, with $m\in\mathbb{N} $. If $f(t)\in
AC^n([0,t])$, then the Caputo fractional derivative exists almost
everywhere on $[0,t]$ and it is represented in the form
\begin{equation}
D_t^{\gamma}f(t)=   J^{m-\gamma}_t D_t^m f(t)=
\frac{1}{\Gamma(m-\gamma)}\int_0^{t}(t-\tau)^{m-\gamma-1}\frac{d^m}{dt^m}f
(\tau) \, \mathrm d\tau, \;\gamma \ne m.
\end{equation}
\end{theorem}
By definition the fractional derivative is a pseudodifferential
operator given by the convolution of the ordinary derivative of the
function with a power law kernel. So the reason why fractional
derivatives introduce a memory formalism becomes evident.

The following properties of fractional derivatives and integrals
(see e.g. \cite{Podl}) will be used in the analysis:
\begin{align}
&D_t^{\gamma} J_t^{\gamma} f(t)= f(t), \quad \gamma> 0,\\
&J_t^{\gamma} D_t^{\gamma} f(t)= f(t)-\sum_{k=0}^{m-1}f^{(k)}(0)\frac{t^k}{k!}, \qquad \gamma>0, \: t>0,\\
&J_t^{\gamma} t^{\delta}= \frac{\Gamma(\delta+1)}{\Gamma(\delta+\gamma+1)}t^{\delta+\gamma} \qquad \gamma>0, \: \delta>-1, \: t>0,\\
&D_t^{\gamma} t^{\delta}=
\frac{\Gamma(\delta+1)}{\Gamma(\delta-\gamma+1)}t^{\delta-\gamma}
\qquad \gamma>0, \: \delta \in (-1,0) \cup (0,+\infty), \: t>0.
\end{align}

\subsection{Invariant subspace method.}
\label{ssec:2}
The invariant subspace method, as introduced by Galaktionov \cite{Gala}, allows to solve exactly
nonlinear equations by separating variables.\\
Recently Gazizov and Kasatkin \cite{Gazizov} suggested its application to nonlinear fractional equations. \\
We recall the main idea of this method: consider a scalar evolution
equation
\begin{equation}\label{pro}
\frac{\partial u}{\partial t}= F[u],
\end{equation}
where $u=u(x,t)$ and
$F[u]\equiv F(u,\partial u/\partial x,\partial^2 u/\partial x^2,\dots,\partial^k u/\partial x^k),\;k\in\mathbb{N}$,
is a nonlinear differential operator. \\
Given $n$ linearly independent functions
$$f_1(x), f_2(x),....,f_n(x),$$
we call $W_n$, the $n$-dimensional linear space
$$W_n=\langle f_1(x), ...., f_n(x)\rangle.$$
This space is called invariant under the given operator $F[u]$, if
$F[y]\in W_n$ for any $y\in W_n$. This means that there exist $n$
functions $\Phi_1, \Phi_2,..., \Phi_n$ such that
$$F[C_1f_1(x)+......C_n f_n(x)]= \Phi_1(C_1,....,C_n)f_1(x)+......+\Phi_n(C_1,....,C_n)f_n(x),$$
where $C_1, C_2, ....., C_n$ are arbitrary constants. \\
Once the set of functions $f_i(x)$  forming the invariant subspace
is given, we search an exact solution of \eqref{pro} in the
invariant subspace in the form
\begin{equation}
u(x,t)=\sum_{i=1}^n u_i(t)f_i(x).
\end{equation}
where $f_i(x)\in W_n$. In this way, we arrive to a system of ODEs.
In many cases this is a simpler problem that allows to find exact
solutions by just separating variables \cite{Gala}.\\
A relevant question in the theory of invariant subspace method is
the following: how to find all the invariant subspaces admitted by a
given differential operator $F[u]$? For the utility of the reader,
we recall that a complete answer to this question is given by the
following Proposition (see \cite[sec. 2.1]{Gala} and
\cite{Gazizov}).
\begin{proposition}
Let $f_1(x),\dots, f_n(x)$ form the fundamental set of solutions of
a linear $n$-th order ordinary differential equation
\begin{equation}\label{prop}
L[y]= y^{(n)}+a_1(x)y^{(n-1)}+\dots+a_{n-1}(x)y'+a_n(x) y=0,
\end{equation}
and $F[y]= F(x,y,y',\dots, y^{(k)})$ a given differential operator
of order $k\leq n-1$, then the subspace $W_n= \langle f_1(x), \dots,
f_n(x)\rangle$ is invariant with respect to $F$ if and only if
\begin{equation}\nonumber
L[F[y]]=0,
\end{equation}
for all solutions $y(x)$ of \eqref{prop}.
\end{proposition}

\section{Third order nonlinear time-fractional dispersive equation.}
\label{sec:3}
In this section we consider the following third order nonlinear time-fractional dispersive equation
\begin{equation}
\frac{\partial^\alpha u}{\partial
t^\alpha}=\frac{\partial^3}{\partial x^3}\left(\frac{u^2}{2}\right),
\quad x\in\mathbb{R}, \; t\geq 0 \label{third},
\end{equation}
where $u=u(x,t)$ and $\alpha\in(0,1]$.\\
The original model equation (corresponding to $\alpha =1$) was
deeply studied in \cite{Gala1}. This equation $K_{\alpha}(0,2,0)$
belongs to the more general class of NDEs \eqref{K}. Here we find an
exact solution of \eqref{third} corresponding to a time-fractional
generalization of the similarity solution discussed in \cite{Gala1}.
First of all, we recall the following useful
\begin{lemma}
For $\alpha\in(0,1/2)\cup(1/2,1)$ equation \eqref{third} admits
\begin{equation}
W_4=\left\lbrace1,x,x^2,x^3\right\rbrace
\end{equation}
 as invariant subspace.\\
\label{lemmainvariantthird}
\end{lemma}
\begin{proof}
By direct calculation, for any
\begin{equation}
g(x)= C_0+C_1 x+C_2 x^2+C_3 x^3\in W_4,
\end{equation}
we have that, being $F[u]= \frac{\partial^3}{\partial
x^3}\left(\frac{u^2}{2}\right)$
\begin{equation}\nonumber
F[g]= 6(C_1C_2+C_0C_3)+12(C_2^2+2C_1C_3)x+60C_2C_3x^2+60 C_3^2
x^3\in W_4,
\end{equation}
as claimed.
\end{proof}
The reason why we must exclude $\alpha =1/2$ and $\alpha=1$ will be clear in the following.\\
Thanks to Lemma \ref{lemmainvariantthird}, it is possible to express
the solution of equation \eqref{third} in the form:
\begin{equation}
u(x,t)=g_0(t)+g_1(t)x+g_2(t)x^2+g_3(t)x^3.
\end{equation}
This leads us to the following system
\begin{equation}
\begin{cases}
d^\alpha g_0/d t^\alpha=6(g_1g_2+g_0g_3),\\
d^\alpha g_1/d t^\alpha=12(g_2^2+2g_1g_3),\\
d^\alpha g_2/d t^\alpha=60g_2g_3,\\
d^\alpha g_3/d t^\alpha=60g_3^2.
\end{cases}
\label{systemthird}
\end{equation}
The solution $u(x,t)$ will be found by solving one by one each equation of system \eqref{systemthird}.\\
The last of equations \eqref{systemthird} is solved by
\begin{equation}
g_3(t)=C_3t^\beta,
\label{ansatz}
\end{equation}
where $C_3$ is a real constant and the exponent $\beta$ can be found
by direct substitution in the equation.\\ It is
\begin{equation}
\frac{d^\alpha}{d t^\alpha}C_3t^\beta=C_3\frac{\Gamma(\beta+1)}{\Gamma(\beta-\alpha+1)}t^{\beta-\alpha}=60C_3^2t^{2\beta},\qquad \beta>-1
\end{equation}
then, \eqref{ansatz} is an effective solution if
\[
\beta=-\alpha,\quad
C_3=\frac{\Gamma(1-\alpha)}{60\Gamma(1-2\alpha)}.
\]
Thus,
\begin{equation}
g_3(t)=\frac{\Gamma(1-\alpha)}{60\Gamma(1-2\alpha)}t^{-\alpha}
\end{equation}
In the same way it is possible to find the solutions of the other equations in \eqref{systemthird}.\\
We obtain that all the solution are in the form
\begin{equation}
g_i(t)=C_it^{-\alpha},\quad C_i\in\mathbb{R},\;i=0,\dots,3.
\end{equation}
Below are reported the solutions $g_i$:
\begin{equation}
\begin{cases}
g_0(t)=C_0t^{-\alpha}\\
g_1(t)=C_1t^{-\alpha}\\
g_2(t)=C_2t^{-\alpha}\\
g_3(t)=C_3 t^{-\alpha},
\end{cases}
\end{equation}
where
\begin{equation}
\begin{cases}\nonumber
C_0=\frac{400}{3}\left(\frac{\Gamma(1-2\alpha)}{60\Gamma(1-\alpha)}\right)^2,\\
C_1= 20\frac{\Gamma(1-2\alpha)}{\Gamma(1-\alpha)}\\
C_2=1\\
C_3= \frac{\Gamma(1-\alpha)}{60\Gamma(1-2\alpha)}.
\end{cases}
\end{equation}
The complete solution $u(x,t)$ results, then,
\begin{equation}\label{soll}
 u(x,t)=\frac{C_0}{t^{\alpha}}+C_1 \frac{x}{t^{\alpha}}+\frac{x^2}{t^{\alpha}}+
 C_3\frac{x^3}{t^{\alpha}}.
\end{equation}
In order to understand the meaning of the above solution, we recall
that the first results about shock and rarefaction waves for NDE
\eqref{third} for $\alpha = 1$ have been discussed in \cite{Gala1}
by analogy with the well known theory for first-order conservation
laws (see e.g. \cite{Evans}). The generalization of their analysis
to the time-fractional case is not completely trivial but suggests
the interpretation of the found solution \eqref{soll} in the
framework of the theory of global similarity solutions of
third-order NDEs.
\begin{remark}
It is important to note that this solution is valid for
$\alpha\in(0,1/2)\cup(1/2,1)$. Indeed for $\alpha =1/2$ and $\alpha
=1$, coefficients appearing in \eqref{soll} can be not defined, due
to the singularity in the Gamma coefficients. The reason why the
found solution is not valid for $\alpha = 1$ is simply given by the
definition of Caputo derivatives. Indeed for an integer $\alpha$ the
Caputo derivative coincides with an ordinary derivative and we must
consider directly the ordinary equation
\begin{equation}\label{thirdog}
\frac{\partial u}{\partial t}=\frac{\partial^3}{\partial
x^3}\left(\frac{u^2}{2}\right), \quad x\in\mathbb{R}, \; t\geq 0,
\end{equation}
deeply studied by Galaktionov and Pohozaev in \cite{Gala1}, where
blowing-up and global similarity solutions have
been considered.\\
 A further remark regards the critical case
$\alpha = 1/2$. From direct calculations we have shown that in this
case the found solution \eqref{soll} is divergent. On the other hand
this critical value of $\alpha$ plays a significant role to
discriminate different regimes. Indeed the sign of the found
solutions is positive for $0<\alpha <1/2$ and negative for
$1/2<\alpha<1$. This non-trivial result should be object of further
research about similarity solutions of time-fractional NDEs.
\end{remark}

\section{Fifth order nonlinear time-fractional dispersive equation.}
\label{sec:4}
In this section we study the following fifth order nonlinear
time-fractional dispersive equation
\begin{equation}
\frac{\partial^\alpha u}{\partial t^\alpha}=\nu\frac{\partial^5
u^2}{\partial x^5}+ \beta\frac{\partial^3 u^2}{\partial
x^3}+\gamma\frac{\partial u^2}{\partial x} \label{quintic},\qquad t\geq 0,\,x\in\mathbb{R},\,\alpha\in(0,1].
\end{equation}
The ordinary case $\alpha=1$ was studied by Dey in \cite{Dey},
where compacton solutions of NDEs were considered.\\
In order to find an exact solution of equation \eqref{quintic}, by the invariant subspace method, we make use of the following Lemma
(see \cite[p. 165]{Gala}).
\begin{lemma}
Equation \eqref{quintic} admits $W_3=\left\lbrace1,\cos x,\sin
x\right\rbrace$ as invariant subspace if and only if
\begin{equation}
16\nu -4\beta+\gamma=0.
\end{equation}
\label{lemmainvariant}
\end{lemma}
The proof of the above Lemma is based on the same arguments used in
Lemma \ref{lemmainvariantthird}. Under the assumption of lemma
\ref{lemmainvariant}, it is possible to express the solution of
equation \eqref{quintic} in the following form:
\begin{equation}
u(x,t)=g_1(t)+g_2(t)\cos x+g_3(t)\sin x,\quad t\geq 0.
\end{equation}
This leads us to following system
\begin{equation}
\begin{cases}
d^\alpha g_1/d t^\alpha=0,\\
d^\alpha g_2/d t^\alpha=\mu g_1g_3,\\
d^\alpha g_3/d t^\alpha=-\mu g_1g_2,
\end{cases}
\label{system1}
\end{equation}
where $\mu=2(\nu-\beta+\gamma)\neq 0$ is a real constant. We assume the condition $\mu\neq 0$ in order to get a non trivial solution of system \eqref{system1}.\\
Now, from the first of equations \eqref{system1}, we get
\begin{equation}
g_1(t)=\mbox{const} =:C
\end{equation}
Then, defining $\bar{\mu}=C\mu$ and by applying the
time-fractional derivative $d^\alpha/dt^\alpha$ to the second of
equations \eqref{system1}, we get the new system
\begin{equation}
\begin{cases}
g_1=C\\
d^\alpha / dt^\alpha d^\alpha g_2/d t^\alpha=-\bar{\mu}^2 g_2\\
d^\alpha g_3/d t^\alpha=-\bar{\mu} g_2.
\end{cases}
\label{system2}
\end{equation}
The main effort in order to solve equation \eqref{system2} is the
solution of the third equation. To this aim we will study the
following Cauchy problem
\begin{equation}
\begin{cases}
\frac{d^\alpha}{d t^\alpha}\frac{d^\alpha}{d t^\alpha}f(t)=-\bar{\mu}^2f(t), \\
f(0)=1\\
\frac{d^\alpha f}{dt^\alpha}(t)\vert_{t=0}=0.
\end{cases}
\label{cauchy}
\end{equation}
The choice of these special initial conditions will be clear in the following.\\
The solution of equation \eqref{cauchy} can be obtained by using the Laplace transform method. To this purpose we
here recall that for $\alpha \in(0,1)$, (see e.g. \cite{Kilbas})
\begin{equation}
\mathcal{L}\left\lbrace\frac{d^\alpha}{d t^\alpha}f(t)\right\rbrace=s^\alpha\tilde{f}(s)-s^{\alpha-1}f(0),
\label{laplace}
\end{equation}
where
\begin{equation}
\mathcal{L}\left\lbrace f(t)\right\rbrace=\tilde{f}(s)=\int_0^\infty
e^{-st}f(t)dt.
\end{equation}
Thus, by setting $h(t):=d^\alpha  f(t)/d t^\alpha$, by formula \eqref{laplace}, we obtain
\begin{equation}
\mathcal{L}\left\lbrace\frac{d^\alpha}{d t^\alpha}\frac{d^\alpha}{d t^\alpha}f(t)\right\rbrace=
s^{2\alpha}\tilde{f}(s)-s^{2\alpha-1},
\end{equation}
where we used both initial conditions of the Cauchy problem \eqref{cauchy}.
We now make the Laplace transform in both terms of the first of equations \eqref{cauchy}, we get
\begin{equation}
s^{2\alpha}\tilde{f}(s)=-\bar{\mu}^2\tilde{f}+s^{2\alpha-1}
\end{equation}
thus
\[
\tilde{f}=\frac{s^{2\alpha-1}}{s^{2\alpha}+\bar{\mu}^2}
\]
whose inverse Laplace transform is given by
\begin{equation}
f(t)=E_{2\alpha ,1}(-\bar{\mu}^2 t^{2\alpha}),\qquad \alpha\in(0,1],
\end{equation}
where $E_{2\alpha ,1}(\cdot)$ is the Mittag-Leffler function
\begin{equation}
E_{2\alpha
,1}(-\bar{\mu}t^{2\alpha})=\sum_{k=0}^\infty\frac{(-1)^k\bar{\mu}^{2k}t^{2\alpha
k}}{\Gamma(2\alpha k+1)}.
\end{equation}
We remind that Mittag-Leffler functions play a fundamental role in the theory of fractional differential equations, see for example \cite{mai,mathaibook,mathaiarticle}.\\
Going back to the system of equations \eqref{system2}, we have that
\begin{equation}
g_2(t)=E_{2\alpha ,1}(-\bar{\mu}^2t^{2\alpha}),
\end{equation}
and, by  substitution in the second equation of \eqref{system2}, we
have
\begin{equation}
g_3(t)=-J^{\alpha}_t\bar{\mu}g_2(t)=-\bar{\mu}t^{\alpha}E_{2\alpha,\alpha+1}\left(-\bar{\mu}^2t^{2\alpha}\right),
\end{equation}
up to an integration additive constant, that we set equal to zero for sake of simplicity.
By the above calculations we are allowed to
write the solution $u(x,t)$ of equation \eqref{quintic} in the form
\begin{equation}\label{sol}
u(x,t)=C+E_{2\alpha ,1}(-\bar{\mu}^2t^{2\alpha})\cos x-
\bar{\mu}t^{\alpha}E_{2\alpha,\alpha+1}\left(-\bar{\mu}^2t^{2\alpha}\right)\sin
x.
\end{equation}
We observe that, for $\alpha =1 $ we retrieve the solution in
\cite[p. 167]{Gala}. Indeed, for $\alpha =1$
\begin{align}
\nonumber & E_{2
,1}(-\bar{\mu}^2t^{2})=\sum_{k=0}^\infty\frac{(-1)^k\bar{\mu}^{2k}t^{2
k}}{2k!}= \cos(\bar{\mu}t),\\
\nonumber & \bar{\mu}t
E_{2,2}\left(-\bar{\mu}^2t^{2}\right)=\sum_{k=0}^\infty\frac{(-1)^k\bar{\mu}^{2k+1}t^{2
k+1}}{(2k+1)!}= \sin(\bar{\mu}t),
\end{align}
so that the solution \eqref{sol} becomes
\begin{equation}
u(x,t)=C+\cos(\bar{\mu}t)\cos x- \sin(\bar{\mu}t)\sin x= C+
\cos(x+\bar{\mu}t),
\end{equation}
that can be written as a compacton (see \cite{Gala}).
\begin{figure}[h!]
\centering
\includegraphics[scale=0.8]{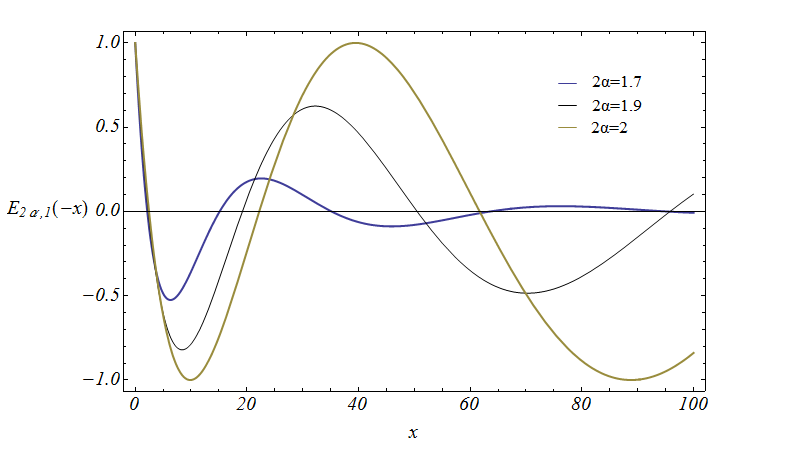}
\caption{Mittag-Leffler function, appearing in \eqref{sol} for
different values of $2\alpha$.} \label{fig:1}
\end{figure}
\newpage
We can now discuss the physical meaning of the introduction of
memory effects by means of time-fractional derivatives in the
fifth order NDE \eqref{quintic}. As it can be seen in Figure \ref{fig:1}, fractional
derivatives introduce damping effects on the time-evolution
depending on the real parameter $\alpha$. We recall that the
Mittag-Leffler function of parameter $\alpha \in (1,2)$ is related
to the so called \textit{fractional oscillation} (see e.g.
\cite{mai,mmai}), that is a damped oscillation where the
damping effect depends on the order of fractionality.

\begin{remark}
We now explain the reason of the choice of the initial conditions in
the Cauchy problem \eqref{cauchy}. It is known that, in general,
\[\frac{d^\alpha}{d t^\alpha}\frac{d^\alpha f}{d t^\alpha}\neq \frac{d^{2\alpha}f}{d t^{2\alpha}}\]
so that we can not write $d^{2\alpha}g_2/d t^{2\alpha}$
instead of $d^\alpha/d t^\alpha d^\alpha g_2/d
t^\alpha$ in the second equation of \eqref{system2}. This
substitution would have been very convenient  to find the
solution without considering an initial condition on the
time-fractional derivative  without a clear physical meaning (a
recent discussion about this point can be find in \cite{hey}). In
light of this fact, in order to find the explicit solution of
\eqref{cauchy} by using the Laplace transform method, we must give
the initial condition on the time-fractional derivative, that we
take null for simplicity.
\end{remark}

\begin{remark}
In the case  $\nu=0,\,\beta=\gamma=1$ in \eqref{quintic} the
time-fractional Rosenau-Hyman equation $K_{\alpha}(2,2)$ is
obtained:
\[
\frac{\partial^\alpha u}{\partial t^\alpha}=\frac{\partial^3
u^2}{\partial x^3}+ \frac{\partial u^2}{\partial x}.
\]
The Rosenau-Hyman equation (see \cite{Rosen}) plays a relevant role
in the theory of solitary waves with compact support. A detailed
study about the fractional Rosenau-Hyman equation traveling wave
solutions should be object of further studies.
\end{remark}

\section{Conclusions and remarks.}
\label{sec:5}
In recent papers time-fractional NDEs have been studied by different
authors with semi-analytical methods. In some cases they find exact
solutions that can be recovered by the invariant subspace method.
For example Odibat in \cite{Odibat} has considered variants of the KdV equation
involving Caputo time-fractional derivatives, such as
\begin{equation}\label{odino}
\frac{\partial^{\alpha}u}{\partial t^{\alpha}}+a \frac{\partial u^2}{\partial x}+\frac{\partial}{\partial x}\left[u\frac{\partial^2 u}{\partial x^2}\right]=0,
\quad t,a>0, \alpha \in (0,1].
\end{equation}
The author finds an explicit solution to \eqref{odino} by means of the homotopy perturbation method, that is
\begin{equation}\label{odi}
u(x,t)= \begin{cases}
\frac{c}{a}\sin^2(\sqrt{a}x, \sqrt{a}ct^{\alpha}, \alpha), \quad |x-ct^{\alpha}|<\frac{\pi}{\mu},\\
0\quad \mbox{otherwise},
\end{cases}
\end{equation}
where $\mu= \frac{\sqrt{a}}{2}$,
\begin{equation}
\sin^2(x,t,\alpha)=1-\cos x\cos(t,\alpha)-\sin x\sin(t, \alpha),
\end{equation}
and
\begin{align}
\nonumber & \cos(t,\alpha)=\sum_{k=0}^{\infty}\frac{(-1)^k t^{2k}}{\Gamma(2k\alpha+1)}=E_{2\alpha,1}\left(-t^2\right)\\
\nonumber & \sin(t,\alpha)=\sum_{k=0}^{\infty}\frac{(-1)^k t^{2k+1}}{\Gamma(2k\alpha+\alpha+1)}=t E_{2\alpha,\alpha+1}\left(-t^2\right).
\end{align}
The solution \eqref{odi} can be found in a direct way by using the
invariant subspace method. Indeed this solution generalizes the
compact wave solution, but it is in separating variable form. This
means that \eqref{odino} admits as invariant subspace
\begin{equation}
W_3=\left\lbrace1,\cos x,\sin x\right\rbrace.
\end{equation}
In the literature about time-fractional NDEs these exact solutions with separating variables are found also in other recent works, even if
with different methods.
In \cite{Odibat1} the author has considered time-fractional $K(n,m)$ equations. Also in this case, some exact results
can be recovered by the invariant subspace method. It can be proved that the fractional equations discussed in \cite{Odibat1} admit as invariant subspace
$W_3=\left\lbrace 1,\cosh \frac{x}{2}, \sinh \frac{x}{2}\right\rbrace$. \\
We can conclude that, considering the literature on time-fractional
NDEs, the invariant subspace method can provide effective and
rigorous tools to find exact solutions to a wide class of nonlinear
fractional equations, avoiding the use of perturbative or
approximate methods. Moreover, from the physical point of view, it
allows to find relevant compacton-like solutions and rarefaction
wave solutions to time-fractional NDEs. The meaning of these
generalized equations is explained in Section \ref{sec:4}, where we
have shown that the role of fractionality is to introduce damping
effects in the evolution of compacton solutions of NDEs.

\baselineskip=0.9\normalbaselineskip

\bibliographystyle{CAIMbibstyle}

\bibliography{biblio}

\end{document}